\documentclass[preprint,nofootinbib,amsmath,amssymb,amsfonts,aps,prd]{revtex4-1}

\usepackage{setspace}
\usepackage{graphicx}
\usepackage{subfigure}
\usepackage[bookmarks=true, pdfstartview={FitH}, bookmarksopen=true, bookmarksopenlevel=1, linktoc=page]{hyperref}
\usepackage{amsthm}
\usepackage{mathrsfs}
\DeclareMathAlphabet{\mathpzc}{OT1}{pzc}{m}{it}

\DeclareMathOperator{\lie}{\pounds}

\newcommand{\be}{\begin{equation}}
\newcommand{\ee}{\end{equation}}
\newcommand{\ba}{\begin{equation}\begin{aligned}}
\newcommand{\ea}{\end{aligned}\end{equation}}

\usepackage{amsthm}
\let\oldproofname=\proofname
\renewcommand{\proofname}{\rm\bf{\oldproofname}:}
\newtheorem*{theorem}{Theorem}
\newtheorem*{corollary}{Corollary}
\newtheorem{lemma}{Lemma}

\begin{document}

\pagestyle{myheadings}

\title{Stationarity of Extremum Entropy Fluid Bodies in General Relativity}
\author{Joshua S. Schiffrin}
\email{schiffrin@uchicago.edu}
\affiliation{Enrico Fermi Institute and Department of Physics \\
  The University of Chicago \\
  5640 S. Ellis Ave., Chicago, IL 60637, U.S.A.}
\date{\today}

\begin{abstract}
We consider perfect fluid bodies (``stars'') in general relativity that are axisymmetric, asymptotically flat, and that admit a maximal hypersurface. We show that configurations that extremize the total entropy at fixed ADM mass, ADM angular momentum, and total particle number are stationary with circular flow. For such stars, this establishes that thermodynamic equilibrium implies dynamic equilibrium.
\end{abstract}

\maketitle

\tableofcontents

\newpage

\section{Introduction} \label{intro}

In the description of self-gravitating fluid bodies---``stars''---in general relativity, the fundamental conserved quantities that characterize the global state of a star can be taken to be the ADM mass $M$, the ADM angular momentum $J$, and the total number of particles $N$. 
The Einstein-perfect fluid equations predict an additional conserved global quantity: the total entropy, $S$. However, $S$ is not on equal footing with $M$, $J$, and $N$. The perfect fluid description is an idealized one in which dissipative processes---i.e., viscosity, heat conduction, and diffusion---cannot occur. In a more realistic fluid description that takes dissipative processes into account, a star that is not in equilibrium would be expected to have $S$ increasing with time, while $M$, $J$ and $N$ remain constant. One would expect the star to eventually settle down to a stationary end-state that is a (local) maximum of $S$ for the given values of $M$, $J$, and $N$.

Suppose, now, that we are given initial conditions for a non-stationary perfect fluid star. Imagine that we ``turn on'' a small amount of dissipation at the initial time, wait a moment, and then ``turn off'' the dissipation. We would expect $S$ to increase during the time that the dissipation is ``on''. Thus, we will have produced a perfect fluid state ``nearby'' the initial state that has higher entropy. This suggests that extrema of $S$ at fixed $M$, $J$, and $N$ must be stationary.
Furthermore, Lindblom \cite{Lindblom} has shown that stationary dissipative fluid stars are necessarily axisymmetric with ``circular flow,'' meaning that the fluid 4-velocity is a linear combination of the stationary and axial Killing vector fields, $t^a$ and $\phi^a$. Thus, if we consider a stationary perfect fluid star that lacks axisymmetry or circular flow and we do our experiment of turning on some dissipation, we would again expect the entropy to increase. This line of reasoning suggests we make the following conjecture: Any perfect fluid star that is an extremum of entropy---i.e., $\delta S=0$ for all linearized solutions that have $\delta M=\delta J = \delta N =0$---is stationary and axisymmetric with circular flow.\footnote{\footnotesize One may furthermore conjecture that any extremum of entropy that is in fact a local maximum---i.e., $\delta^2 S < 0$ for all perturbations that keep $M$, $J$, and $N$ fixed to both first and second order---is dynamically stable. In fact, under the \emph{assumptions} of stationarity, axisymmetry, and circular flow, this was shown to be true in \cite{GSW} (and it was earlier shown in the case of static, spherically symmetric radiation stars by \cite{SWZ}, and in the case of static, spherically symmetric barotropic stars by \cite{Roupas}).\medskip} 

In this paper we give a proof that any \emph{axisymmetric} perfect fluid star that admits a maximal (i.e., zero trace of the extrinsic curvature) hypersurface and that is at an extremum of entropy (at fixed $M$, $J$, and $N$) is stationary with circular flow. In other words, we will prove the above conjecture for stars that are assumed from the beginning to be axisymmetric, and that admit a maximal hypersurface. As we will discuss in section \ref{assumptions}, the latter assumption is not believed to exclude any physically reasonable spacetimes.

The question of proving the above conjecture has been previously addressed in the literature in the case where spherical symmetry is assumed, in which case $J$ is removed from the problem. In particular, Sorkin, Wald, and Zhang \cite{SWZ} considered radiation fluid enclosed in a spherical box. They argued---but as we will describe below, did not prove---that any configuration that admits a maximal hypersurface and that extremizes $S$ at fixed $M$ (there is no notion of $N$ for a radiation fluid) is static\footnote{\footnotesize Other related results have been proven in the spherically symmetric case. In \cite{SWZ} it was shown for configurations of a radiation fluid  that admit a time-symmetric (i.e., zero extrinsic curvature) hypersurface that being an extremum of entropy is sufficient \emph{and} necessary for being static, and Gao \cite{Gao} extended the ``sufficient'' part of this result to allow for a general 2-parameter equation of state (the ``necessary'' part holds only for 1-parameter equations of state). Fang and Gao \cite{FangGao} obtained some related results without assuming spherical symmetry, but they assumed from the start that the spacetime metric is static and that the redshifted temperature is uniform.\medskip}. Our stationarity proof is modeled after their staticity argument, so we will review it here. Afterwards, we will explain what the obstacles are in extending their method to the case of axisymmetric rotating stars, and how we overcome these obstacles.

The basic idea behind the argument of Sorkin, Wald, and Zhang is as follows: If there is any bulk fluid flow, it should be possible to construct a perturbation that slows down this flow and converts its energy into thermal energy, thereby raising the entropy. 
Constructing perturbed initial data that does this is nontrivial because the constraint equations must be satisfied. 
They start by assuming the existence of a maximal hypersurface $\Sigma$. The momentum constraint on such a surface is
\be\label{momconstraint}
D^j K_{ij}=-8\pi P_i,
\ee 
and the Hamiltonian constraint is
\be\label{hamconstraint}
R(h)=16 \pi \varepsilon + K_{ij}K^{ij},
\ee
where $h_{ij}$ is the spatial metric, $K_{ij}$ is the extrinsic curvature, $P_i$ is the fluid momentum density, and $\varepsilon$ is the fluid energy density (as viewed by an observer moving normal to $\Sigma$).
They construct their perturbation on $\Sigma$ by first doing a ``momentum scaling'' which scales down $K_{ij}$ and $P_i$ by a constant factor while keeping $h_{ij}$ fixed. As $h_{ij}$ is unchanged, so is $M$. 
The momentum scaling doesn't change the momentum constraint \eqref{momconstraint}, but it decreases the term $K_{ij}K^{ij}$ in the Hamiltonian constraint \eqref{hamconstraint}. They compensate for this by increasing $\varepsilon$, which can be shown to increase the entropy density of the radiation fluid. Thus, if the configuration is an extremum of $S$ at fixed $M$, we must have $K_{ij}=0$---i.e., $\Sigma$ must be a time-symmetric hypersurface. Additionally, by \eqref{momconstraint}, the fluid momentum density must be zero, implying that the fluid 4-velocity is normal to $\Sigma$. But existence of one maximal hypersurface $\Sigma$ implies \cite{Cantoretal} a unique foliation by maximal hypersurfaces in a neighborhood of $\Sigma$, so it follows that a neighborhood of $\Sigma$ can be foliated by time-symmetric hypersurfaces. But this means that the neighborhood, and therefore the entire spacetime, is static. 

Where their argument fails to be a proof is in the step where they compensate for the decrease in $K_{ij}K^{ij}$ with an increase in $\varepsilon$. Although this works to keep the Hamiltonian constraint satisfied \emph{inside} the fluid body, it does not work \emph{outside} the body where $\varepsilon=0$.\footnote{Their argument can be made into a proof by ``fixing up'' their perturbation via a linearized conformal transformation of the initial data, as we will do below in the general axisymmetric case.} 

The basic idea of our stationarity proof in the axisymmetric case, in parallel with the basic idea of the above staticity argument, is as follows: If the fluid flow has any ``non-circular'' part---i.e., a part perpendicular to the axial Killing vector field $\phi^a$---it should be possible to construct a perturbation that slows down this non-circular flow and converts its energy into thermal energy, thereby raising the entropy. Only the circular part of the flow contributes to the star's angular momentum, so such a perturbation should not change $J$. 

We will construct a perturbation analogous to the one of \cite{SWZ} described above to use in the axisymmetric case with general equation of state to prove our stationarity theorem. There are three main obstacles to doing this: (i) In addition to $M$, the angular momentum, $J$, must now be held fixed under the perturbation. (ii) The particle number, $N$, must also be held fixed. (iii) The perturbation must satisfy the constraint equations everywhere on $\Sigma$, including \emph{outside} the star where $\varepsilon=0$.

Obstacle (ii) (keeping $N$ fixed) is straightforward to overcome, since, as we will show, one has the freedom to keep the particle number density fixed while scaling down $P_i$ and increasing the entropy density. Overcoming obstacle (i) (keeping $J$ fixed) is more nontrivial, however. We will do this---following methods developed in \cite{tphipaper} for axisymmetric vacuum spacetimes---by using the axisymmetry to decompose $h_{ij}$, $K_{ij}$ and $P_i$ into ``axial'' and ``polar'' parts with respect to the axial Killing vector field. Our analog of the ``momentum scaling'' will be defined to scale down the axial part of $h_{ij}$ and the polar parts of $K_{ij}$ and $P_i$ by a constant factor, while keeping the polar part of $h_{ij}$ and the axial parts of $K_{ij}$ and $P_i$ fixed. As we will show, this keeps $M$ and $J$ unchanged. 

To overcome obstacle (iii) we need a way to ``fix up'' changes in the Hamiltonian constraint other than by changing $\varepsilon$ to compensate. Following \cite{SudarskyWald} and \cite{tphipaper}, we will do so via a linearized conformal transformation of the initial data, i.e., a linearized version of the Lichnerowicz method \cite{Lichnerowicz}. By choosing the conformal factor appropriately we can make sure that the linearized constraints are satisfied everywhere while also controlling $M$ and $J$ as needed to allow us to prove our stationarity theorem. 

Our proof can be summarized as follows: Given the spacetime of an axisymmetric perfect fluid star that admits a maximal hypersurface $\Sigma$, we construct perturbed initial data on $\Sigma$ in two parts: The first perturbation---which we call ``the momentum scaling perturbation''---scales down the axial part of $h_{ij}$ and the polar parts of $K_{ij}$ and $P_i$ while keeping the particle number density and entropy density fixed. This is then ``fixed up'' with a linearized conformal transformation so that the linearized constraints are satisfied and so that $J$, $N$, and $S$ remain fixed. It is shown that $M$ will be decreased by this perturbation unless the axial part of $h_{ij}$ and the polar parts of $K_{ij}$ and $P_i$ vanish. The second perturbation---which we call ``the heating perturbation''---increases the entropy density in an arbitrarily chosen region in the interior of the star while keeping all other variables fixed, and again, the constraints are fixed up via a linearized conformal transformation. It is shown that $S$ and $M$ will both be increased by this perturbation, while $N$ and $J$ remain fixed. Thus, a linear combination of these two perturbations can be found that increases $S$ at fixed $M$, $N$, and $J$, unless the axial part of $h_{ij}$ and the polar parts of $K_{ij}$ and $P_i$ vanish. Thus, if the star is an extremum of $S$ at fixed $M$, $N$, and $J$, these parts must vanish.

Stationarity is then proven as follows. We show that the vanishing of the axial part of $h_{ij}$ implies the existence of a $\phi$-reflection isometry of $\Sigma$ that reverses the direction of the axial Killing vector field. Then, the fact that the polar parts of $K_{ij}$ and $P_i$ vanish implies that $\Sigma$ is a ``$t$-$\phi$''-reflection symmetric surface. But it then follows that a neighborhood of $\Sigma$ can be foliated by $(t$-$\phi)$-reflection symmetric hypersurfaces, which implies that the neighborhood, and therefore the entire spacetime, is stationary. Circularity of the fluid flow follows from the vanishing of the polar part of $P_i$.

In section \ref{perfectfluid}, we briefly review the Einstein-perfect fluid system. In section \ref{assumptions}, we spell out our assumptions on the perfect fluid star spacetimes being considered. In section \ref{manifoldoforbits}, we show how to decompose the initial data into axial and polar parts, and we write the constraint equations as equations on the ``manifold of orbits'' of the axial Killing vector field. In section \ref{perturbation}, we construct the perturbation that will be used in our proof, in two parts. In section \ref{proof}, we state and prove our theorem. Finally, in section \ref{blackholes}, we discuss some possible generalizations and extensions.

Our index notational conventions are as follows: Lower case Latin indices from the early alphabet $(a,b,c,\dots)$ denote abstract spacetime indices, while lower case Latin indices from mid-alphabet $(i,j,k,\dots)$ denote abstract spatial indices. We do not distinguish notationally between tensors on a (3-dimensional) spatial hypersurface and tensors on the (2-dimensional) manifold of orbits of the axial Killing vector field. Lower case Greek indices $(\mu,\nu,\dots)$ denote coordinate labels and coordinate components of tensor fields. 

Additionally, three distinct metrics appear in the paper. The (4-dimensional, Lorentzian) spacetime metric is denoted $g_{ab}$, and its associated derivative operator is denoted $\nabla_a$. The (3-dimensional, Riemannian) spatial metric is denoted $h_{ij}$, its associated derivative operator is denoted $D_i$, and its scalar curvature is denoted $R(h)$. The metric on the (2-dimensional) manifold of orbits is denoted $\mu_{ij}$, its associated derivative operator is denoted $\mathcal D_i$, and its scalar curvature is denoted $\mathcal R(\mu)$.

\section{The Einstein-Perfect Fluid Formalism}\label{perfectfluid}

Here we review the Einstein-perfect fluid formalism; for more comprehensive overviews see, e.g., \cite{Friedmanbook, GSW}. The local field variables on a spacetime manifold $\mathcal M$ that describe the Einstein-perfect fluid system consist of the spacetime metric $g_{ab}$, the particle number density $n$, the entropy per particle $s$, and the fluid 4-velocity $u^a$ satisfying $u^a u_a = -1$. The stress energy tensor has the form 
\be
T_{ab} = \left( \rho + p \right)u_a u_b + p g_{ab}
\ee
where the energy density in the rest frame of the fluid, $\rho$, is determined by a prescribed equation of state,
\be
\rho = \rho(n,s),
\ee
and the pressure, $p$, is given by the Gibbs-Duhem relation
\be\label{gibbsduhem}
p=-\rho+\mu n+T s n,
\ee
where
\be
T\equiv \frac{1}{n}\frac{\partial \rho}{\partial s}
\ee
is the temperature and
\be
\mu \equiv \frac{\partial\rho}{\partial n}-Ts
\ee
is the chemical potential. The Einstein-perfect fluid field equations consist of Einstein's equation, 
\be
G_{ab}=8\pi T_{ab},
\ee
and conservation of the particle number current,
\be\label{ncons}
\nabla_a\left(nu^a\right) = 0.
\ee
The latter together with conservation of stress energy (which is, of course, implied by Einstein's equation) implies conservation of the entropy current,
\be\label{scons}
\nabla_a\left(snu^a\right) = 0.
\ee

The equation of state is required to be such that for all allowed $(n,s)$ we have
\be\label{eosconds}
\rho \geq 0, \qquad p\geq 0, \qquad T>0, \qquad 0\leq c_s^2 \leq 1,
\ee
where 
\be\label{cs}
c_s^2 \equiv \left(\frac{dp}{d\rho}\right)_s = \frac{\partial p/\partial n}{\partial\rho/\partial n} 
\ee
is the square of the speed of sound in the fluid. Under these conditions, the Einstein-perfect fluid system is well posed; see e.g., \cite{CBYork}. Additional assumptions on the second derivatives of the equation of state should hold so that the fluid is locally thermodynamically stable---see \cite{GSW}---but they don't play a role here. 

As we are interested in configurations that represent stars, we will be considering only asymptotically flat solutions in which $n$ has compact spatial support. The fundamental conserved global quantities are then taken to be the ADM mass, $M$, the ADM angular momentum, $J$, the total number of particles, 
\be\label{N}
N = -\int_\Sigma   n u^a \nu_a \sqrt{h}\,d^3x,
\ee
and the total entropy,
\be\label{S}
S = -\int_\Sigma   s n u^a \nu_a \sqrt{h}\,d^3x,
\ee
where $\Sigma$ is any Cauchy hypersurface, $\nu^a$ is the unit future-directed normal to $\Sigma$, and $\sqrt{h}\,d^3x$ is the induced volume element on $\Sigma$. By \eqref{ncons} and\eqref{scons}, $N$ and $S$ are conserved, i.e., independent of the choice of $\Sigma$.

Initial data on a Cauchy hypersurface is given (see, e.g., \cite{CBYork}) by $\Psi\equiv(h_{ij}, K_{ij}, n, s, u^i)$, where $h_{ij}$ is the induced metric, $K_{ij}$ is the extrinsic curvature, and $u^i$ is the projection of $u^a$ tangent to the surface. The quantity $u^a\nu_a$ appearing in \eqref{ncons} and\eqref{scons} can be written in terms of initial data by noting that $g_{ab}u^a u^b=-1$ gives
\be
u^a\nu_a = -\sqrt{1+u^2},
\ee
where we have defined
\be
u^2 \equiv h_{ij}u^i u^j.
\ee

The initial value constraint equations are
\be\label{constraint1}
D^j K_{ij} - D_i K^j_{\phantom jj} +8\pi P_i = 0
\ee
and
\be\label{constraint2}
-R(h) +K_{ij}K^{ij}-\left(K^i_{\phantom ii}\right)^2 + 16 \pi \varepsilon =0,
\ee
where the fluid momentum density, $P_c=-\nu^a h_c^{\phantom c b} T_{ab}$, is given by
\be\label{Jvrelation}
P_i =  \left(\rho+p\right)\sqrt{1+u^2\,}\, u_i ,
\ee
and the energy density, $\varepsilon=\nu^a \nu^b T_{ab}$, is given by
\be\label{varepsilon}
\varepsilon =  \left(\rho+p\right) \left( 1+u^2 \right) - p.
\ee
Clearly, on a maximal hypersurface---i.e., one with $K^i_{\phantom ii} =0$---the constraints reduce to \eqref{momconstraint} and \eqref{hamconstraint}.

\section{Assumptions}\label{assumptions}

Here, a \emph{perfect fluid star spacetime}, $(\mathcal M, g_{ab}, n, s, u^a)$, is a globally hyperbolic, topologically $\mathbb R^4$, asymptotically flat solution of the Einstein-perfect fluid field equations such that $n$ has compact spatial support, and such that the ``physical reasonableness'' conditions on the equation of state, \eqref{eosconds}, hold.

We will consider only perfect fluid star spacetimes that obey the following additional assumptions: (i) The spacetime is \emph{axisymmetric}, meaning that the symmetry group contains a $U(1)$ subgroup with spacelike orbits; equivalently, that the spacetime admits a spacelike Killing vector field with closed orbits, $\phi^a$, such that  we also have $\lie_\phi n=\lie_\phi s=\lie_\phi u^a=0$. (ii) The $U(1)$ symmetry group acts trivially, in the sense we will define at the end of this section. (iii) The spacetime admits a maximal ($K^i_{\phantom ii}=0$) asymptotically flat Cauchy hypersurface, $\Sigma$.  

While axisymmetry is a genuine restriction (we will discuss the possibility of relaxing it in section \ref{blackholes}), assumption (iii) on the other hand is thought to exclude only ``unphysical'' spacetimes containing white hole regions, as in Brill's explicit example, \cite{Brill}, of an asymptotically flat spacetime that does not admit a maximal hypersurface. (His example consists a ``pocket'' of expanding dust glued into the left wedge of the maximally extended Schwarzschild spacetime.) In fact, Bartnik \cite{Bartnik} has proven the existence of maximal hypersurfaces in asymptotically flat spacetimes that satisfy a ``uniformity condition'' in the interior\footnote{\footnotesize Existence of maximal hypersurfaces has been proven for asymptotically flat \emph{stationary} spacetimes---see \cite{ChruscielWald,BartnikChrusciel}---but obviously we do not want to assume stationarity here.\medskip}.

For asymptotically flat solutions of Einstein's equation with matter that satisfies the strong energy condition, the following two related results are known: (a) If the spacetime admits a maximal hypersurface $\Sigma$, then there exists a unique foliation by maximal hypersurfaces in a neighborhood of $\Sigma$; see \cite{Cantoretal}. (b) If the spacetime is axisymmetric, the axial Killing vector field, $\phi^a$, is tangent to any maximal hypersurface $\Sigma$. This follows from axisymmetry, asymptotic flatness, and the aforementioned uniqueness of the maximal foliation: If $\phi^a$ was not tangent to each maximal hypersurface, one would obtain an inequivalent maximal foliation by performing a rotation\footnote{A precise version of this argument was given in the vacuum black hole case in \cite{tphipaper}, which is easily generalized to the case where matter satisfying the strong energy condition is present.}. Since the perfect fluid matter under consideration here satisfies the strong energy condition---i.e., $(\rho+p)\geq0$ and $(\rho+3p)\geq0$, which follow from \eqref{eosconds}---assumptions (i) and (iii) imply that we have a unique maximal foliation in a neighborhood of $\Sigma$, and that the axial Killing vector field $\phi^a$ is tangent to each hypersurface in this foliation. 

As $\phi^a$ is tangent to $\Sigma$, we will henceforth denote it as a spatial vector, $\phi^i$, on $\Sigma$. Let $\mathcal A\subset\Sigma$ denote the rotation axis, i.e., the set of points in $\Sigma$ at which $\phi^i$ vanishes, and let $\widetilde\Sigma = \Sigma\setminus\mathcal A$. Our assumption, (ii), of trivial action of the $U(1)$ symmetry group is that we can write
\be\label{trivialaction}
\widetilde\Sigma = \mathcal O \times U(1),
\ee
where $\mathcal O$ is a 2-dimensional manifold, called the \emph{manifold of orbits of $\phi^i$}. Written this way, $\widetilde\Sigma$ has the structure of a trivial principal fiber bundle with fiber group $U(1)$ and base manifold $\mathcal O$. 
We believe that under our assumptions that space is topologically $\mathbb R^3$ and asymptotically flat, one automatically has $\widetilde\Sigma=\mathcal O \times U(1)$. We have not attempted to prove this, but if true it would make assumption (ii) redundant.

\section{Axial/Polar Decomposition, the Manifold of Orbits, and the Constraints}\label{manifoldoforbits}

Many of the constructions in this section may be viewed as a specialization of the ``manifold of orbits'' formalism of Geroch \cite{Geroch} to the case of a Riemannian manifold. Our exposition and notation closely follows \cite{tphipaper}, where the case of axisymmetric vacuum spacetimes in higher dimensions is considered.

On the maximal hypersurface $\Sigma$, we can uniquely decompose the fluid momentum density into its parts parallel and perpendicular to the axial Killing vector field, $\phi^i$, via
\be\label{Jdecomp}
P_i= \mathcal P \phi_i+\mathpzc p_i,
\ee
where $\phi^i\mathpzc p_i=0$. Similarly, the extrinsic curvature can be decomposed via
\be\label{Kdecomp}
K_{ij} = 2\mathcal K_{(i} \phi_{j)}+\kappa \phi_i \phi_j + k_{ij},
\ee
where $k_{ij}$ and $\mathcal K_i$ are orthogonal to $\phi^i$.  
We can decompose any axisymmetric tensor field $T^{ij\cdots}_{\phantom{ij\cdots}kl\cdots}$ on $\Sigma$ into it's ``axial'' and ``polar'' parts via such a decomposition, where the axial part consists of any terms in the decomposition that contain an odd number of copies of $\phi^i$, and the polar part consists of any terms in the decomposition that contain an even number of copies of $\phi^i$. Thus, the axial part of the fluid momentum density is $\mathcal P \phi_i$, while the polar part is $\mathpzc p_i$. Similarly, the axial part of the extrinsic curvature is $2\mathcal K_{(i} \phi_{j)}$, while the polar part is $k_{ij}+\kappa \phi_i \phi_j$. Under the action of a $\phi$-reflection isometry---if one exists---the axial part of an axisymmetric tensor will reverse sign, while the polar part will remain unchanged.

As explained in \cite{tphipaper}, any axisymmetric tensor field on $\widetilde\Sigma$ that is orthogonal in each of its indices to $\phi^i$ can be projected to the manifold or orbits, $\mathcal O$. The tensor fields $k_{ij}$, $\kappa$, $\mathcal K_i$, $\mathpzc p_i$, and $\mathcal P$ are examples of tensors that project to $\mathcal O$. As mentioned in section \ref{intro}, we will not distinguish notationally between such tensors on $\widetilde\Sigma$ and their projection to $\mathcal O$. We would like to write the initial value constraint equations, \eqref{momconstraint} and \eqref{hamconstraint}, completely in terms of quantities that project to $\mathcal O$. To do this, we need, in addition to the decompositions \eqref{Jdecomp} and \eqref{Kdecomp} of $P_i$ and $K_{ij}$, a way to write the spatial metric, $h_{ij}$, and the quantities in the constraints associated with it---i.e., $D_i$ and $R(h)$---in terms of objects defined on $\mathcal O$. Following \cite{tphipaper}, we can encode all of the information contained in $h_{ij}$ in a triplet, $\left(\Phi, A_i, \mu_{ij}\right)$, of tensor fields that project to $\mathcal O$. We now define these quantities.

First, $\Phi$ is defined as the norm-squared of $\phi^i$,
\be
\Phi \equiv \phi_i \phi^i,
\ee
which, as it vanishes in $\Sigma$ only on $\mathcal A$, is positive everywhere on $\widetilde\Sigma$. Next, $\mu_{ij}$ is defined as
\be
\mu_{ij} = h_{ij} -\frac{1}{\Phi}\phi_i \phi_j,
\ee
which, when projected to $\mathcal O$, is a (Riemannian) metric on $\mathcal O$. Finally, to define $A_i$, we first consider the object $\Phi^{-1} \phi_i$. As explained in \cite{tphipaper}, $\Phi^{-1} \phi_i$ defines a connection on the principal fiber bundle $\widetilde\Sigma = \mathcal O \times U(1)$. Choose a smooth global cross-section, $\mathcal S$, of $\widetilde\Sigma$, and let $A_i$ be the axisymmetric 1-form field on $\widetilde\Sigma$ whose pullback to $\mathcal S$ is equal to the pullback of $\Phi^{-1} \phi_i$ to $\mathcal S$ and is such that $\phi^i A_i=0$. Then $A_i$ contains all the nontrivial information contained in the connection $\Phi^{-1} \phi_i$, but is ``gauge-dependent'' in that it depends on a choice of cross-section $\mathcal S$.
Given local coordinates $(x^1, x^2)$ on $\mathcal S$, one can locally construct coordinates $(\varphi, x^1, x^2)$ on $\widetilde\Sigma$ by using $(x^1, x^2)$ to label the orbits of $\phi^i$ and by choosing  $(\partial/\partial\varphi)^i=\phi^i$. In these coordinates, the spatial metric, $h_{ij}$, takes the form
\be\label{metricform}
ds^2 = \Phi d\varphi^2 + 2\Phi A_\mu d\varphi dx^\mu +\left(\mu_{\mu\nu}+\Phi A_\mu A_\nu\right)dx^\mu dx^\nu.
\ee
Thus, to show that $h_{ij}$ admits a $\phi$-reflection isometry, we must show that the gauge freedom in choosing $\mathcal S$ (equivalently, a coordinate redefinition of the form $\varphi\to\varphi+f(x^\mu)$) can be used to set $A_\mu \to 0$, so that $h_{ij}$ has a ``block diagonal'' form. 

The \emph{curvature 2-form}, $F_{ij}$, is defined as the curvature of the connection $\Phi^{-1}\phi_i$, i.e., 
\be\label{curv1}
F_{ij} = 2 D_{[i} \left(\Phi^{-1}\phi_{j]}\right).
\ee
It is easily verified that $F_{ij}$ is orthogonal to $\phi^i$, and therefore projects to $\mathcal O$. As shown in \cite{tphipaper}, it follows that 
\be\label{curv2}
F_{ij} = 2 D_{[i} A_{j]}.
\ee
The Frobenius condition for local 2-surface orthogonality of $\phi^i$ is that the ``twist'' of $\phi^i$ vanishes:
\be
\phi_{[i}D_{j} \phi_{k]} = 0.
\ee
But from \eqref{curv1} we have
\be
\phi_{[i}D_{j} \phi_{k]}=\phi_{[i}D_{j}\left(\Phi \Phi^{-1} \phi_{k]}\right) = \frac{1}{2}\Phi\phi_{[i}F_{jk]}
\ee
so the Frobenius condition is equivalent to $F_{ij}=0$. 

As shown in \cite{tphipaper}, both $A_i$ and $F_{ij}$ can be smoothly extended to the rotation axis, $\mathcal A$, and therefore the relation
\be\label{FArelation}
F_{ij} = 2 D_{[i} A_{j]}
\ee
holds everywhere on $\Sigma$. By assumption, $\Sigma$ is simply connected, and hence if $F_{ij}$ vanishes then there exists a smooth function $\chi$ such that $A_i=D_i \chi$. In that case we can set $A_i \to 0$ by displacing $\mathcal S$ by $-\chi$ along the orbits of $\phi^i$; equivalently, we can redefine $\varphi\to\varphi+\chi$ and thereby set $A_\mu=0$ in \eqref{metricform}. 
Thus, $F_{ij}=0$ is necessary and sufficient for the existence of a $\phi$-reflection isometry of $h_{ij}$.

We now have the ingredients in place to write the constraint equations, \eqref{momconstraint} and \eqref{hamconstraint}, as equations that project to the manifold of orbits. The derivation differs from \cite{tphipaper} only by the addition of contributions from the fluid matter and in the specialization to a single $\phi^i$, so we will not repeat the derivation here. The axial part (i.e., the component along $\phi^i$) of the momentum constraint is
\be\label{c1}
\frac{1}{\sqrt{\Phi}} \mathcal D^i \left( \Phi^{3/2} \mathcal K_i \right)+8\pi \Phi \mathcal P=0,
\ee
where $\mathcal D_i$ is the derivative operator associated with $\mu_{ij}$. The polar part of the momentum constraint is
\be\label{c2}
\frac{1}{\sqrt{\Phi}} \mathcal D^j \left( \sqrt{\Phi} k_{ij} \right) - \frac{1}{2}\kappa \mathcal D_i \Phi - \Phi \mathcal K^j F_{ij} +8\pi \mathpzc p_i=0,
\ee
and the Hamiltonian constraint is
\be\label{c3}
-\mathcal R(\mu)+\frac{1}{\Phi}\mathcal D^2 \Phi-\frac{1}{2\Phi^2}\left(\mathcal D_i \Phi\right)\left(\mathcal D^i \Phi\right)+2\Phi\mathcal K^i \mathcal K_i+ \left[\frac{1}{4}\Phi F_{ij}F^{ij} +k^{ij}k_{ij}+\Phi^2 \kappa^2 \right]+16\pi\varepsilon=0,
\ee
where $\mathcal R(\mu)$ is the scalar curvature of $\mu_{ij}$.

\section{The Perturbation}\label{perturbation}

Perturbed initial data on $\Sigma$ that generates a solution to the linearized Einstein-perfect fluid equations on $\mathcal M$ consists of the collection of tensor fields 
\be\label{initialdata}
\delta \Psi \equiv \left(\delta h_{ij}, \delta K_{ij}, \delta n, \delta s, \delta u^i\right)
\ee
on $\Sigma$, subject only to the condition that these quantities satisfy the linearized constraint equations, i.e., the linearization of \eqref{constraint1} and \eqref{constraint2}. 

We will write our perturbation of interest, $\delta \Psi$, as a linear combination of two perturbations, $\delta_\text{\tiny MS} \Psi$ (``the momentum scaling perturbation'') and $\delta_\text{\tiny H} \Psi$ (``the heating perturbation''), each of which satisfies the linearized constraint equations individually. These two perturbations can be given the following interpretation: $\delta_\text{\tiny MS} \Psi$ scales down the ``non-circular'' (i.e., polar) part of the fluid motion, thereby decreasing the total energy, $M$. $\delta_\text{\tiny H} \Psi$ uses this energy to ``heat up'' an arbitrarily chosen portion of the star, thereby increasing $S$. Neither of these perturbations change $J$ or $N$.

Each of these two perturbations will in turn be constructed in two parts: a first part which satisfies the momentum constraint but not the Hamiltonian constraint, and a second part consisting on a linearized conformal transformation designed to ``fix up'' the Hamiltonian constraint. 

The first part of $\delta_\text{\tiny MS} \Psi$ will be designed to scale down the quantities $(F_{ij}, k^{ij},\kappa, \mathpzc p_i)$ by a constant factor while keeping the quantities $(\Phi, \mu_{ij}, \mathcal K^i, \mathcal P, \mathfrak N, s)$ fixed, where
\be
\mathfrak N \equiv \sqrt{h} n \sqrt{1+u^2}.
\ee
(Note: It is not immediately obvious that perturbed initial data, \eqref{initialdata}, can be chosen to accomplish this; we will show that it can.)
By inspection, this keeps unchanged both the axial and polar parts of the momentum constraint, \eqref{c1} and \eqref{c2},
and, since $\mathfrak N$ and $s\,\mathfrak N$ are the integrands of \eqref{N} and \eqref{S}, it keeps $N$ and $S$ fixed.
As we will show, it also does not change $M$ or $J$. Furthermore, it does not change the first 4 terms of the Hamiltonian constraint, \eqref{c3}, but it decreases the term in square brackets and, as we will show, it decreases the last term, $16\pi\varepsilon$. Thus, the first part of $\delta_\text{\tiny MS} \Psi$ will not satisfy the linearized Hamiltonian constraint. The second part of $\delta_\text{\tiny MS} \Psi$ will fix this---it will be a linearized conformal transformation chosen to make the linearized Hamiltonian constraint hold while keeping $J$, $N$, and $S$ fixed. We will show that this necessarily decreases $M$ unless each of the quantities $(F_{ij}, k^{ij},\kappa, \mathpzc p_i)$ vanish everywhere on $\Sigma$.

The first part of $\delta_\text{\tiny H}\Psi$ will be chosen to increase $s$ in an arbitrarily chosen region in the interior of the star (while keeping it fixed everywhere else) while keeping the quantities $(h_{ij}, K_{ij}, P_i, \mathfrak N)$ fixed everywhere (again, we will show that this can be done). This satisfies the linearized momentum constraint but not the linearized Hamiltonian constraint, since, as we will show, it increases $\varepsilon$ in the region where the entropy density is increased. Again, we will fix this up with a linearized conformal perturbation, and the resulting perturbation will be shown to increase $M$ and $S$ while keeping $J$ and $N$ fixed. 

Thus, unless the quantities $(F_{ij}, k^{ij},\kappa, \mathpzc p_i)$ 
vanish everywhere on $\Sigma$, we will be able to find a linear combination of $\delta_\text{\tiny MS} \Psi$ and $\delta_\text{\tiny H} \Psi$ that increases $S$ while keeping $M$, $J$, and $N$ fixed.

\subsection{The ``Momentum Scaling''  Perturbation}

We specify $\delta_\text{\tiny MS} \Psi$ as a sum of two parts:
\be\label{delta1Psi}
\delta_\text{\tiny MS} \Psi = \delta_\text{\tiny MS}^\text{\tiny(1)}\Psi + \delta_\text{\tiny MS}^\text{\tiny(2)} \Psi.
\ee
$\delta_\text{\tiny MS}^\text{\tiny(1)} \Psi$ will be the perturbation that scales down $(F_{ij}, k^{ij},\kappa, \mathpzc p_i)$, while $\delta_\text{\tiny MS}^\text{\tiny(2)} \Psi$ will be the linearized conformal perturbation. We will now specify these perturbations, starting with $\delta_\text{\tiny MS}^\text{\tiny(1)}\Psi$.

First we specify the perturbed gravitational variables, $\left(\delta_\text{\tiny MS}^\text{\tiny(1)} h_{ij}, \delta_\text{\tiny MS}^\text{\tiny(1)} K_{ij}\right)$, via
\ba \label{pert1a}
\delta_\text{\tiny MS}^\text{\tiny(1)} h_{ij} &= -2\phi_{(i}A_{j)}\\
\delta_\text{\tiny MS}^\text{\tiny(1)} K_{ij} &= -k_{ij}-\kappa \phi_i\phi_j -2\Phi\kappa \phi_{(i}A_{j)}-2\Phi \mathcal K_{(i}A_{j)},
\ea
or equivalently, in terms of the variables introduced in section \ref{manifoldoforbits}, 
\ba\label{pert1alt}
\delta_\text{\tiny MS}^\text{\tiny(1)} \Phi = 0, \qquad \delta_\text{\tiny MS}^\text{\tiny(1)} \mu_{ij}&=0, \qquad \delta_\text{\tiny MS}^\text{\tiny(1)} A_i = -A_i,\\
\delta_\text{\tiny MS}^\text{\tiny(1)} k_{ij} = -k_{ij}, \qquad \delta_\text{\tiny MS}^\text{\tiny(1)} \kappa&=-\kappa, \qquad \delta_\text{\tiny MS}^\text{\tiny(1)} \mathcal K_i =0.
\ea
Note that we are considering $\phi^i$ to be fixed as a vector field, i.e., it is unchanged under any perturbation we consider here. In particular, $\delta_\text{\tiny MS}^\text{\tiny(1)} \phi_i = \delta_\text{\tiny MS}^\text{\tiny(1)}(h_{ij} \phi^j) = \phi^j\delta_\text{\tiny MS}^\text{\tiny(1)}h_{ij} = - \Phi A_i$, which accounts for the last two terms in \eqref{pert1a}. 
It follows from \eqref{FArelation} that 
\be
\delta_\text{\tiny MS}^\text{\tiny(1)} F_{ij} = - F_{ij}.
\ee

Next, we would like to specify the perturbed fluid variables, $\left( \delta_\text{\tiny MS}^\text{\tiny(1)} n, \delta_\text{\tiny MS}^\text{\tiny(1)} s,\delta_\text{\tiny MS}^\text{\tiny(1)} u^i\right)$, so that the following three conditions hold: (i) the perturbed fluid momentum density satisfies
\be
\delta_\text{\tiny MS}^\text{\tiny(1)} P_i = -\mathpzc p_i -\Phi \mathcal P A_i,
\ee
or equivalently, in terms of the variables introduced in the previous section,
\be
\delta_\text{\tiny MS}^\text{\tiny(1)} \mathpzc p_i = -\mathpzc p_i, \qquad \delta_\text{\tiny MS}^\text{\tiny(1)} \mathcal P=0;
\ee
(ii) the perturbed particle number density on $\Sigma$ is zero, i.e.,
\be \label{deltan1}
\delta_\text{\tiny MS}^\text{\tiny(1)} \mathfrak N = \delta_\text{\tiny MS}^\text{\tiny(1)} \left(\sqrt{h}\, n\sqrt{1+u^2} \right) = \sqrt{h}\,\delta_\text{\tiny MS}^\text{\tiny(1)} \left(\, n\sqrt{1+u^2} \right) = 0,
\ee
where the last equality holds because \eqref{pert1a} gives $\delta_\text{\tiny MS}^\text{\tiny(1)}\sqrt{h}=0$;
and (iii) the perturbed entropy density on $\Sigma$ is zero, i.e.,
\be \label{deltas1}
\delta_\text{\tiny MS}^\text{\tiny(1)}\left(s \,\mathfrak N\right) = 0,
\ee
which, given (ii), is equivalent to $\hat \delta_\text{\tiny MS} s=0$. These three conditions uniquely determine $\left( \delta_\text{\tiny MS}^\text{\tiny(1)} n, \delta_\text{\tiny MS}^\text{\tiny(1)} s, \delta_\text{\tiny MS}^\text{\tiny(1)} u^i\right)$, and formulas for these can be calculated explicitly from the definition of $P_i$, \eqref{Jvrelation}, and from the formula for $p$, \eqref{gibbsduhem}, which in particular implies that
\be
\frac{\partial}{\partial n} \left(\frac{\rho+p}{n}\right) =\frac{1}{n}\frac{\partial p}{\partial n}= c_s^2 \frac{\rho+p}{n^2},
\ee
where $c_s^2$ was defined in \eqref{cs}. We find that these formulas are
\begingroup
\addtolength{\jot}{1.3em}
\ba\label{pert1b}
\delta_\text{\tiny MS}^\text{\tiny(1)} s &= 0
\\
\delta_\text{\tiny MS}^\text{\tiny(1)} n &=\frac{n\mathpzc p^i \mathpzc p_i}{\bigl[1+(1-c_s^2)u^2\bigr](1+u^2)(\rho+p)^2} 
= \frac{n\mu^{ij}u_i u_j}{\bigl[1+(1-c_s^2)u^2\bigr]}
\\
\delta_\text{\tiny MS}^\text{\tiny(1)} u^i &= -\frac{\delta_\text{\tiny MS}^\text{\tiny(1)} n }{n} c_s^2u^i +\frac{(-\mathpzc p^i+A_j \mathpzc p^j \phi^i)}{(\rho+p)\sqrt{1+u^2}}=\left[-\frac{\delta_\text{\tiny MS}^\text{\tiny(1)} n }{n} c_s^2\delta^i_{\phantom ij} -\mu^i_{\phantom ij} +A_j \phi^i\right]u^j,
\ea
\endgroup
where we have used $\mathpzc p^i =(\rho+p)\sqrt{1+u^2} \mu^i_{\phantom ij} u^j$ to write these in a way that explicitly shows that these define a smooth perturbation (i.e., that there is no problem near the edge of the star where $(\rho+p)\to 0$). From \eqref{varepsilon} we can calculate the perturbation to $\varepsilon$ under $\delta_\text{\tiny MS}^\text{\tiny(1)}\Psi$, and we find
\be
\delta_\text{\tiny MS}^\text{\tiny(1)}\varepsilon = -\frac{\mathpzc p^i \mathpzc p_i}{(1+u^2)(\rho+p)} 
= -(\rho+p)\mu^{ij}u_i u_j,
\ee
which is manifestly non-positive.

By inspection, $\delta_\text{\tiny MS}^\text{\tiny(1)} \Psi$ produces no change in the momentum constraints, \eqref{c1} and \eqref{c2}, and produces no change in the first four terms of the Hamiltonian constraint, \eqref{c3}. However it produces changes in the other terms, giving
\be
\delta_\text{\tiny MS}^\text{\tiny(1)} \left( -R(h)+K_{ij} K^{ij} +16\pi\varepsilon\right) = -\mathcal Q,
\ee
where 
\be\label{Q}
\mathcal Q \equiv \frac{1}{2}\Phi F_{ij}F^{ij} +2k^{ij}k_{ij}+2\Phi^2 \kappa^2 + \frac{16\pi\mathpzc p^i \mathpzc p_i}{(1+u^2)(\rho+p)}.
\ee
Clearly, $\mathcal Q$ is nonnegative, and is zero if and only if $F_{ij}=k_{ij}=\kappa=\mathpzc p_i=0$.

Next we specify the linearized conformal perturbation, $\delta_\text{\tiny MS}^\text{\tiny(2)} \Psi$. Our construction follows \cite{tphipaper}, but differs in that we must also specify the perturbed fluid variables and in that we do not have a black hole horizon present. The perturbed gravitational variables are given by 
\ba \label{pert2a}
\delta_\text{\tiny MS}^\text{\tiny(2)} h_{ij} &= \psi h_{ij}\\
\delta_\text{\tiny MS}^\text{\tiny(2)} K_{ij} &=-\frac{1}{2} \psi K_{ij},
\ea
where the scalar function $\psi$ will be specified below. The perturbed fluid variables are specified by requiring
\ba\label{conformalfluid}
\delta_\text{\tiny MS}^\text{\tiny(2)} \left(\sqrt{h}\,P_i\right)&=0\\
\delta_\text{\tiny MS}^\text{\tiny(2)} \mathfrak N&=0\\
\delta_\text{\tiny MS}^\text{\tiny(2)}  s&=0.
\ea
The (unique) values of $\left(\delta_\text{\tiny MS}^\text{\tiny(2)} n, \delta_\text{\tiny MS}^\text{\tiny(2)} s,\delta_\text{\tiny MS}^\text{\tiny(2)} u^i\right)$ that give \eqref{conformalfluid} are
\begingroup
\addtolength{\jot}{.7em}
\ba \label{pert2b}
\delta_\text{\tiny MS}^\text{\tiny(2)} s &= 0\\
\delta_\text{\tiny MS}^\text{\tiny(2)} n &=- \frac{3+2u^2}{2\bigl[1+(1-c_s^2)u^2\bigr]}\psi n \\
\delta_\text{\tiny MS}^\text{\tiny(2)} u^i &= \frac{c_s^2(3+4u^2)-2(1+u^2)}{2\bigl[1+(1-c_s^2)u^2\bigr]}\psi u^i,
\ea
\endgroup
and from these, one finds
\be\label{tildedelta1varepsilon}
\delta_\text{\tiny MS}^\text{\tiny(2)} \varepsilon = -\frac{1}{2}(3+4 u^2)(\rho+p)\psi.
\ee
$\delta_\text{\tiny MS}^\text{\tiny(2)} \Psi$ satisfies the momentum constraint,
\be
\frac{1}{\sqrt{h}}\, \delta_\text{\tiny MS}^\text{\tiny(2)} \left(\sqrt{h}D^j K_{ij} + 8\pi \sqrt{h}P_i \right) =0,\\
\ee 
but not the Hamiltonian constraint:
\ba
\delta_\text{\tiny MS}^\text{\tiny(2)} \left( -R(h)+K_{ij} K^{ij} +16\pi\varepsilon\right) &= 2 D^2 \psi + R(h)\psi -3 K_{ij} K^{ij}\psi + 16 \pi \delta_\text{\tiny MS}^\text{\tiny(2)} \varepsilon\\
&= 2 D^2 \psi -2 K_{ij} K^{ij}\psi - 8 \pi\left[ (\rho+3p)+2u^2(\rho+p)\right]\psi.
\ea
We choose $\psi$ to be the solution of
\be\label{lich0}
- D^2 \psi +\left[K_{ij} K^{ij} + 4 \pi(\rho+3p)+8\pi u^2(\rho+p)\right]\psi = -\frac{1}{2}\mathcal Q
\ee 
with the boundary condition $\psi\to0$ at infinity. As the quantity in square brackets multiplying $\psi$ is nonnegative, by standard arguments \cite{CantorBrill,ChruscielDelay} there exists a unique solution of \eqref{lich0} satisfying this boundary condition. It immediately follows that the perturbation $\delta_\text{\tiny MS} \Psi=\delta_\text{\tiny MS}^\text{\tiny(1)}\Psi+\delta_\text{\tiny MS}^\text{\tiny(2)}\Psi$ satisfies the linearized constraint equations. 

\begin{lemma}\label{lemma1}
The perturbation $\delta_\text{\tiny\em MS}\Psi$ (see \eqref{pert1a}, \eqref{pert1b}, \eqref{pert2a}, \eqref{pert2b}, \eqref{lich0}) has $\delta_\text{\tiny\em MS} N=\delta_\text{\tiny\em MS} S= \delta_\text{\tiny\em MS} J = 0$. Furthermore, $\delta_\text{\tiny\em MS} M \leq 0$ with equality holding if and only if $\mathcal Q=0$ everywhere on $\Sigma$.
\end{lemma}

\begin{proof}
The total particle number, \eqref{N}, and the total entropy, \eqref{S}, are
\ba
N &= \int_\Sigma \mathfrak N\, d^3x\\
S &= \int_\Sigma s \,\mathfrak N\, d^3x.
\ea
From \eqref{deltan1} and \eqref{deltas1} we have $\delta_\text{\tiny MS}^\text{\tiny(1)} N = \delta_\text{\tiny MS}^\text{\tiny(1)} S =0$. From \eqref{conformalfluid} we have $\delta_\text{\tiny MS}^\text{\tiny(2)} N=\delta_\text{\tiny MS}^\text{\tiny(2)} S=0$. Thus, $\delta_\text{\tiny MS} N=\delta_\text{\tiny MS} S=0$.

The ADM angular momentum is 
\be\label{JADM}
J=-\frac{1}{8\pi} \int_{S_\infty} K_{ij}\phi^i r^j\sqrt{\gamma}\,d^2 x,
\ee
where $S_\infty$ denotes a sphere with radius taken to infinity, $r^i$ is its outward pointing unit normal, and $\sqrt{\gamma}\,d^2 x$ is its induced volume element.
Only the axial part of $K_{ij}$ contributes to this integral, and this is unchanged under $\delta_\text{\tiny MS}^\text{\tiny(1)}\Psi$ (see \eqref{pert1alt}), so we have $\delta_\text{\tiny MS}^\text{\tiny(1)} J = 0$ \footnote{\footnotesize Another way to see this is to note that under the assumption of axisymmetry we may write \mbox{$J= \int_{\Sigma} \sqrt{h}\phi^i P_i =\int_{\Sigma} \sqrt{h}\Phi \mathcal P$}, thus, since $\delta_\text{\tiny MS}^\text{\tiny(1)} \sqrt{h}=\delta_\text{\tiny MS}^\text{\tiny(1)} \Phi=\delta_\text{\tiny MS}^\text{\tiny(1)}\mathcal P=0$, we have $\delta_\text{\tiny MS}^\text{\tiny(1)} J=0$.\medskip}. Since $\psi\to0$ at infinity we also have $\delta_\text{\tiny MS}^\text{\tiny(2)} J=0$. Thus, $\delta_\text{\tiny MS} J = 0$.

The perturbed ADM mass is given by
\be\label{MADM}
\delta M = \frac{1}{16\pi}\int_{S_\infty}r^i h^{jk} \left( D_k \delta h_{ij} - D_i \delta h_{jk} \right)\sqrt{\gamma}\,d^2 x.
\ee
Since $\delta_\text{\tiny MS}^\text{\tiny(1)} h_{ij}$ is purely axial, we have $\delta_\text{\tiny MS}^\text{\tiny(1)} M=0$ \footnote{\footnotesize One might question whether it makes sense to talk about the ADM mass or angular momentum of a perturbation that doesn't satisfy the constraints. For convenience, we take the attitude here that $\delta M$ and $\delta J$ are defined by the formulas \eqref{MADM} and \eqref{JADM}, independent of whether the constraints are satisfied. For our final perturbations $\delta_\text{\tiny MS}\Psi$ and $\delta_\text{\tiny H} \Psi$ that do satisfy the constraints, the values given by these formulas are then equal to the ADM values.\medskip}, so, substituting $\delta_\text{\tiny MS}^\text{\tiny(2)} h_{ij} = \psi h_{ij}$, we find
\be
\delta_\text{\tiny MS} M=\delta_\text{\tiny MS}^\text{\tiny(2)} M = -\frac{1}{8\pi}\int_{S_\infty}(r^i   D_i \psi) \sqrt{\gamma}\,d^2 x .
\ee
To evaluate this, let $\xi$ be the solution of
\be\label{xieq}
- D^2 \xi +\left[K_{ij} K^{ij} + 4 \pi(\rho+3p)+8\pi u^2(\rho+p)\right]\xi = 0
\ee
with boundary condition $\xi\to1$ at infinity. By the strong maximum principle \cite{Gilbarg} $\xi$ is strictly positive. Furthermore, by \eqref{lich0}, we have
\be
D^i \left(\xi D_i \psi - \psi D_i \xi\right) = \frac{1}{2}\xi\mathcal Q.
\ee
Integrating this equation over $\Sigma$ and using that $\psi\to0$ and $\xi\to1$ at infinity, we find
\be
\int_{S_\infty}(r^i D_i \psi)\sqrt{\gamma}\,d^2 x = \frac{1}{2}\int_\Sigma \xi\mathcal Q \sqrt{h}\,d^3 x,
\ee
so that we have
\be\label{deltaM}
\delta_\text{\tiny MS} M = -\frac{1}{16\pi}\int_\Sigma \xi \mathcal Q \sqrt{h}\,d^3 x.
\ee
Since $\xi>0$ and $\mathcal Q \geq 0$, we have $\delta_\text{\tiny MS} M \leq 0$ with equality holding if and only if $\mathcal Q = 0$ identically.
\end{proof}

\subsection{The ``Heating'' Perturbation}

As we did with $\delta_\text{\tiny MS} \Psi$, we specify $\delta_\text{\tiny H}\Psi$ as a sum of two parts:
\be
\delta_\text{\tiny H} \Psi = \delta_\text{\tiny H}^\text{\tiny(1)}\Psi+\delta_\text{\tiny H}^\text{\tiny(2)}\Psi.
\ee
$ \delta_\text{\tiny H}^\text{\tiny(1)}\Psi$ will increase $s$ in an arbitrarily chosen region in the interior of the star while keeping $\mathfrak N$ and $P_i$ unchanged, and it will not change $(h_{ij}, K_{ij})$. Again, $\delta_\text{\tiny H}^\text{\tiny(2)}\Psi$ will be a linearized conformal perturbation chosen to make the perturbation satisfy the linearized Hamiltonian constraint while not changing $J$.
We will now specify these perturbations, starting with $ \delta_\text{\tiny H}^\text{\tiny(1)}\Psi$.

Firstly, we specify
\ba
 \delta_\text{\tiny H}^\text{\tiny(1)} h_{ij} &= 0\\
 \delta_\text{\tiny H}^\text{\tiny(1)} K_{ij} &= 0.
\ea
Next, we specify the perturbed fluid variables $\left(  \delta_\text{\tiny H}^\text{\tiny(1)} n,  \delta_\text{\tiny H}^\text{\tiny(1)} s, \delta_\text{\tiny H}^\text{\tiny(1)} u^i\right)$ by the conditions
\ba\label{pert2conds}
 \delta_\text{\tiny H}^\text{\tiny(1)} s &= f \\
 \delta_\text{\tiny H}^\text{\tiny(1)} \mathfrak N &= 0\\
 \delta_\text{\tiny H}^\text{\tiny(1)} P_i &= 0,
\ea 
where $f$ is an arbitrarily chosen smooth function on $\Sigma$ that is positive in some arbitrarily chosen region in the interior of the star and zero everywhere else. The unique values of $\left(  \delta_\text{\tiny H}^\text{\tiny(1)} n, \delta_\text{\tiny H}^\text{\tiny(1)} u^i\right)$ that reproduce \eqref{pert2conds} are
\ba
 \delta_\text{\tiny H}^\text{\tiny(1)} n &= \frac{n u^2 \left(nT+\frac{\partial p}{\partial s}\right)f}{(\rho+p)\bigl[1+(1-c_s^2)u^2\bigr]}\\
 \delta_\text{\tiny H}^\text{\tiny(1)} u^i &= -\frac{(1+u^2) \left(nT+\frac{\partial p}{\partial s}\right)f}{(\rho+p)\bigl[1+(1-c_s^2)u^2\bigr]}u^i.
\ea
Note that since we have chosen $f$ to be supported only in the interior of the star where $(\rho+p)>0$, these define a smooth perturbation. We can calculate $ \delta_\text{\tiny H}^\text{\tiny(1)} \varepsilon$ from these formulas, and we find
\be
 \delta_\text{\tiny H}^\text{\tiny(1)} \varepsilon = n T f,
\ee
which is positive, by our assumption of positive temperature \eqref{eosconds}.

Clearly, $ \delta_\text{\tiny H}^\text{\tiny(1)}\Psi$ satisfies the linearized momentum constraint, but it violates the linearized Hamiltonian constraint, giving
\be
 \delta_\text{\tiny H}^\text{\tiny(1)} \left( -R(h)+K_{ij} K^{ij} +16\pi\varepsilon\right) = 16 \pi  \delta_\text{\tiny H}^\text{\tiny(1)} \varepsilon =16\pi  n T f.
\ee

Now we specify the linearized conformal perturbation, $\delta_\text{\tiny H}^\text{\tiny(2)} \Psi$. The perturbed gravitational variables are given by
\ba \label{pert2a2}
\delta_\text{\tiny H}^\text{\tiny(2)} h_{ij} &= \widetilde\psi h_{ij}\\
\delta_\text{\tiny H}^\text{\tiny(2)} K_{ij} &=-\frac{1}{2} \widetilde\psi K_{ij},
\ea
where $\widetilde\psi$ is chosen to be the solution of
\be\label{lich}
- D^2 \widetilde\psi +\left[K_{ij} K^{ij} + 4 \pi(\rho+3p)+8\pi u^2(\rho+p)\right]\widetilde\psi = 8\pi  n T f
\ee 
with boundary condition $\widetilde\psi\to0$ at infinity. The perturbed fluid variables, $\left(\delta_\text{\tiny H}^\text{\tiny(2)} n, \delta_\text{\tiny H}^\text{\tiny(2)} s,\delta_\text{\tiny H}^\text{\tiny(2)} u^i\right)$ are chosen to make \eqref{conformalfluid} hold with $\delta_\text{\tiny MS}^\text{\tiny(2)} \to \delta_\text{\tiny H}^\text{\tiny(2)}$, which again gives \eqref{pert2b} and \eqref{tildedelta1varepsilon}  with $\delta_\text{\tiny MS}^\text{\tiny(2)} \to \delta_\text{\tiny H}^\text{\tiny(2)}$ and $\psi \to \widetilde\psi$.

With these choices, it follows that $\delta_\text{\tiny H} \Psi =  \delta_\text{\tiny H}^\text{\tiny(1)} \Psi+\delta_\text{\tiny H}^\text{\tiny(2)} \Psi$ satisfies the linearized constraint equations.

\begin{lemma}\label{lemma2}
The perturbation $\delta_\text{\tiny\em H}\Psi$ has $\delta_\text{\tiny\em H} N= \delta_\text{\tiny\em H} J = 0$. Furthermore, it has $\delta_\text{\tiny\em H} S> 0$ and $\delta_\text{\tiny\em H} M> 0$.
\end{lemma}

\begin{proof}
We have $ \delta_\text{\tiny H}^\text{\tiny(1)} \mathfrak N=\delta_\text{\tiny H}^\text{\tiny(2)} \mathfrak N = 0$, so $\delta_\text{\tiny H} N = 0$. 
We have $ \delta_\text{\tiny H}^\text{\tiny(1)} K_{ij}=0$ and $\widetilde\psi \to 0$ at infinity, so from \eqref{JADM} we have $\delta_\text{\tiny H} J = 0$. We have $\delta_\text{\tiny H}^\text{\tiny(2)} s=0$, so 
\be
\delta_\text{\tiny H} S = \delta_\text{\tiny H}^\text{\tiny(1)} S= \int_\Sigma f \mathfrak N \, d^3x>0.
\ee 
Finally, we have $ \delta_\text{\tiny H}^\text{\tiny(1)} h_{ij}=0$, so $ \delta_\text{\tiny H}^\text{\tiny(1)} M=0$, and by the same argument given in the proof of lemma \ref{lemma1}, but with the replacement $(-\mathcal Q) \to 16 \pi n T f$, we have
\be
\delta_\text{\tiny H} M = \delta_\text{\tiny H}^\text{\tiny(2)} M = \int_\Sigma \xi n T f \sqrt{h}\, d^3x> 0.
\ee
\end{proof}

\section{Stationarity Theorem}\label{proof}

\begin{theorem}
Let $(\mathcal M, g_{ab}, n, s, u^a)$ be a perfect fluid star spacetime satisfying the assumptions of section \ref{assumptions} with axial Killing vector field $\phi^a$. Suppose that $\delta S=0$ for all solutions of the linearized Einstein-perfect fluid equations that have $\delta M=\delta J=\delta N=0$. Then $(\mathcal M, g_{ab}, n, s, u^a)$ is stationary with circular flow---i.e., there exists an asymptotically timelike Killing vector field $t^a$ such that $\lie_t\phi^a=\lie_t n =\lie_t s=\lie_t u^a=0$, and such that $u^a$ is a linear combination of $t^a$ and $\phi^a$.
\end{theorem}

\begin{proof}
Let $\Sigma$ be a maximal Cauchy surface, and consider a perturbation having the initial data 
\be
\delta \Psi=\delta_\text{\tiny H} \Psi + \lambda\, \delta_\text{\tiny MS} \Psi,
\ee
where $\delta_\text{\tiny MS} \Psi$ and $\delta_\text{\tiny H} \Psi$ were constructed in the previous section, and $\lambda$ is a constant. If $\mathcal Q$ is nonzero anywhere on $\Sigma$ then by lemma \ref{lemma1} we have $\delta_\text{\tiny MS} M < 0$, so we can choose $\lambda = \delta_\text{\tiny H} M / |\delta_\text{\tiny MS} M|$ which makes $\delta M=0$. But by lemmas \ref{lemma1} and \ref{lemma2} we have $\delta N = \delta J = 0$ and $\delta S > 0$, which is a contradiction. Thus, $\mathcal Q=0$ everywhere on $\Sigma$, where $\mathcal Q$ was defined by \eqref{Q}. Consequently, we have
\be
F_{ij}=k_{ij}=\kappa=\mathpzc p_i = 0
\ee
everywhere on $\Sigma$.

As shown in section \ref{manifoldoforbits}, $F_{ij}=0$ implies that there exists a $\phi$-reflection isometry $i_\Sigma:\Sigma\to\Sigma$ that reverses the sign of the axial Killing vector field (i.e., $i_\Sigma^*( \phi^i) = -\phi^i$). Then $k_{ij}=\kappa=\mathpzc p_i = 0$ implies that $K_{ij}$ and $P_i$ are odd under $i_\Sigma$ (i.e., $i_\Sigma^*(K_{ij})=-K_{ij}$ and $i_\Sigma^*(P_i)=-P_i$). It follows from the definition of $P_i$, \eqref{Jvrelation}, that $u^i$ is odd under $i_\Sigma$. 

Now consider the ``$t\text{-}\phi$''-reflection, $i:\mathcal M\to \mathcal M$, obtained by mapping a point $p$ lying at a proper time $\tau$ along a normal geodesic starting at $s\in\Sigma$ to the point $q$ lying at proper time $-\tau$ along a normal geodesic starting at $i_\Sigma (s)$. Since $i$ maps the initial data $(h_{ij}, K_{ij}, n,s, u^i)$ on $\Sigma$ into itself, by uniqueness of Cauchy evolution, $i$ is a symmetry of the solution.

Next consider the foliation by maximal hypersurfaces $\Sigma(t)$ in a neighborhood of $\Sigma=\Sigma(0)$; by \cite{Cantoretal}, such a foliation is guaranteed to exist and is unique. By repeating the above argument for each surface in the foliation, we have a $(t\text{-}\phi)$-reflection symmetry, $i_t$, about each $\Sigma(t)$. We can choose the parameter $t$ so that\footnote{The uniqueness of the maximal foliation allows us to make this choice. To see this explicitly, one can choose $t$ to be the proper time at which a particular normal geodesic starting at $\Sigma$ passes through $\Sigma(t)$.} $i_t[\Sigma(t+s)]=\Sigma(t-s)$. The desired 1-parameter group of time-translation symmetries, $I_t:\mathcal M\to \mathcal M$, is then obtained via composition of reflections:
\be
I_t = i_{t/2}\circ i_0
\ee
(it can easily be seen that $I_t$ takes $\Sigma(s)$ to $\Sigma(t+s)$). We have thus shown stationarity in a neighborhood of $\Sigma$, but, by uniqueness of Cauchy evolution, this implies stationarity on all of $\mathcal M$.

Finally, circularity of the fluid flow is shown by noting that the fact that $u^i$ is odd under $i_{\Sigma}$ means that $u^i$ is parallel to $\phi^i$. Thus, $u^a$ is a linear combination of $\phi^a$ and the unit normal vector field, $\nu^a$, to $\Sigma(t)$. But $t^a$ must also be a linear combination of $\nu^a$ and $\phi^a$ (having a component tangent to $\Sigma$ and orthogonal to $\phi^i$ would contradict the fact that $t^a$ is odd under $i$), thus $u^a$ is a linear combination of $t^a$ and $\phi^a$.
\end{proof}

The following corollary was proven in \cite{GSW}\footnote{\footnotesize The uniformity of $\Omega$, $\widetilde T$, and $\widetilde \mu$ was shown earlier by \cite{KatzManor}, and the first law then follows from a result of \cite{Iyer}.\medskip}:

\begin{corollary}
Under the hypothesis of the above theorem, the quantities $\Omega$, $\widetilde T\equiv \left|V\right|T$, and $\widetilde \mu\equiv \left|V\right|\mu$ are uniform, where 
\be
u^a = \frac{t^a+\Omega \phi^a}{\left|V\right|}. 
\ee
In other words, axisymmetric thermal equilibrium stars are rigidly rotating with uniform redshifted temperature and uniform redshifted chemical potential. Furthermore, the \emph{first law of thermodynamics},
\be
\delta M = \widetilde T \delta S + \widetilde\mu \delta N + \Omega \delta J,
\ee
holds for all solutions of the linearized Einstein-perfect fluid equations.
\end{corollary}

\section{Generalizations and Extensions}\label{blackholes}

An obvious generalization we would like to be able to make is to relax the assumption of axisymmetry in our stationarity theorem.
However, having an axial Killing vector field was essential in our method of proof, since without it we don't know how to decompose the fluid flow into a part that does not contribute to the angular momentum (polar) and a part that does contribute (axial), so that we can construct a perturbation that slows down the fluid flow without changing the angular momentum. 
One might therefore try to argue directly that there are no non-axisymmetric extrema of entropy. At first glance, the following argument would seem to do this, but, as we will explain, it does not actually work:

For an axisymmetric star, perturbations consisting purely of gravitational radiation cannot change $J$ to first order (see, e.g., the discussion surrounding equation (63) in \cite{GSW}). But for a non-axisymmetric star, they can. Thus, one might think that for a non-axisymmetric rotating star, one could always construct a perturbation that has $\delta S>0$ with $\delta M=\delta J=\delta N=0$ by: (i) slowing the rotation, (ii) using most of the energy thereby gained to heat up the star, and (iii) putting the rest of the gained energy and all of the gained angular momentum into a gravitational wave-packet very far away from the star. By putting the wave-packet arbitrarily far away, it can carry the given amount of angular momentum while having arbitrarily small energy, thereby guaranteeing that there will be enough energy left over to heat up the star. 

However, it has been shown \cite{CorvinoSchoen} that there exist asymptotically flat non-axisymmetric solutions of the initial value constraints that are exactly axisymmetric outside of a compact region. For such solutions, a wave packet that is too far away from the star would be in the axisymmetric region and therefore could not change $J$ to first order. Thus, the above argument does not work, and currently we do not see how to proceed without the assumption of axisymmetry.

Other generalizations include considering (i) spacetimes of dimension $d>4$ that are axisymmetric---meaning that if there are $N\leq\lfloor\frac{d-1}{2}\rfloor$ nonzero independent angular momenta $\left(J^1,\dots,J^N\right)$, then there exist $N$ mutually commuting axial Killing vector fields $\left(\phi_1^{\phantom1 i},\dots,\phi_N^{\phantom N i}\right)$ associated with them---and (ii) asymptotically anti-de Sitter spacetimes. As shown in \cite{tphipaper}, as long as the action of the rotational symmetry group is trivial in the sense discussed at the end of section \ref{assumptions} (in the higher dimensional case this is believed to be a genuine restriction, as discussed there), one can construct a smooth perturbation on $\Sigma$ that scales the curvature 2-form and the polar parts of the extrinsic curvature, where now these are objects that have extra ``internal'' indices corresponding the the multi-dimensional space of rotational symmetries. Thus, the constructions of section \ref{perturbation} should generalize straightforwardly to the case of higher dimensional stars having trivial action of the rotational isometries. Furthermore, as noted in \cite{tphipaper}, including a \emph{negative} cosmological constant does not affect our ability to solve the linearized Hamiltonian constraint \eqref{lich0} of a conformal perturbation with a given source, so our theorem should generalize straightforwardly to stars in asymptotically AdS spacetimes.

Finally, we discuss the analog of the conjecture of section \ref{intro} for the case of black hole spacetimes. 
Namely, that if the area of a cross section of the event horizon---which plays the role of entropy in black hole thermodynamics---is an extremum at fixed $M$ and $J$, then the black hole is stationary and axisymmetric\footnote{\footnotesize Note that the converse---i.e., stationary black holes are extrema of event horizon area at fixed $M$ and $J$---follows directly from the first law of black hole mechanics.\medskip}. If the area of a cross section of the event horizon is an extremum, then the expansion of the null geodesic generators must vanish at that cross section, because otherwise an infinitesimal diffeomorphism along the null generators (i.e., a ``time translation'') would increase the area, contradicting it being an extremum. But by Hawking's area theorem \cite{Hawking}, the expansion must be nonnegative everywhere, so the Raychaudhuri equation implies that the expansion and shear vanish everywhere on the horizon (in particular, this means the area is constant). It might then follow by arguments similar to those used to prove the rigidity theorem \cite{Hawking, FRW} that there exists a Killing vector field normal to the horizon, but we don't have a proof of this.

The above argument---which uses the ``global methods'' of the area and rigidity theorems---obviously has a very different character than the ``scaling of initial data'' method we used to prove our fluid star stationarity theorem. This can be understood by the fact that the event horizon is a globally defined quantity and cannot easily be located from initial data. However, \emph{apparent} horizons are defined in terms of initial data on a hypersurface, so one might try use our methods to prove stationarity of axisymmetric extrema of apparent horizon area on a maximal hypersurface with a trapped region. As the apparent horizon and event horizon coincide in the stationary case, this seems like a reasonable conjecture to make. But to give it the interpretation of ``extrema of black hole entropy are stationary," one would have to argue that apparent horizon area, rather than event horizon area, should play the role of black hole entropy. Such arguments have been made, e.g., in \cite{Hiscock, Nielsen}, but we will not take a position here. In any case, it appears that it is not actually possible to prove such a conjecture using our methods for the following reason: The momentum scaling and conformal perturbations move the location of the apparent horizon within the maximal slice, and ``moving it back'' (via an infinitesimal diffeomorphism) generally changes the area in an uncontrolled way. 

\section*{Acknowledgments}

I am very grateful to Robert Wald for his excellent guidance and for many extremely helpful discussions. I also thank Stephen Green, Travis Maxfield, and Kartik Prabhu for useful discussions. This research was supported in part by NSF grant PHY 12-02718 to the University of Chicago.
 
\bibliography{mybib}

\end{document}